\documentclass[unsortedaddress,jmp]{revtex4-1}

\usepackage{amsmath}
\usepackage{amsfonts}
\usepackage{amsthm}
\usepackage{amssymb}
\usepackage{graphicx}
\usepackage{hyperref}
\usepackage{color}

	\newcommand{\ncd}{\newcommand}
	\ncd{\mrm}    {\mathrm}
	\ncd{\beq} {\begin{equation}}
	\ncd{\eeq} {\end{equation}}

	\def\d{{\rm d}}

	\newtheorem{prop}{Proposition}
	\newtheorem{corolary}{Corolary}
	\newtheorem{lemma}{Lemma}

\begin{document}

	\title{The conformal metric structure of Geometrothermodynamics}

\author{Alessandro Bravetti$^{1,2}$, Cesar S. Lopez-Monsalvo$^2$, Francisco Nettel$^3$ and Hernando Quevedo$^{1,2,4}$}

\affiliation{$^1$Dipartimento di Fisica and ICRA, Universit\`a di Roma La Sapienza, Piazzale Aldo Moro 5, I-00185 Rome, Italy\\
$^2$Instituto de Ciencias Nucleares, Universidad Nacional Aut\'onoma de M\'exico, AP 70543, M\'exico DF 04510, Mexico \\
$^3$ Departamento de F\'\i sica, Facultad de Ciencias, 
Universidad Nacional Aut\'onoma de M\'exico,
AP 50542, M\'exico DF 04510, Mexico\\
$^4$Instituto de Cosmologia, Relatividade e Astrofisica ICRA - CBPF\\
Rua Dr. Xavier Sigaud, 150, CEP 22290-180, Rio de Janeiro, Brazil
}

\email{bravetti@icranet.org,cesar.slm@correo.nucleares.unam.mx,fnettel@ciencias.unam.mx,quevedo@nucleares.unam.mx}

\date{\today}

	\begin{abstract}
	We present a thorough analysis on the invariance of the most widely used metrics in the Geometrothermodynamics (GTD) programme. We centre our attention 
	in the invariance of the curvature of the space of equilibrium states under a change of fundamental representation.
Assuming that the systems under consideration can be described by a fundamental relation which is a homogeneous function of a definite order, we demonstrate that  such invariance is only compatible with total Legendre transformations in the present form of the programme. We give the explicit form of a metric which is invariant under total Legendre transformations and whose induced metric produces a curvature which is independent of the fundamental representation. Finally, we study a generic system with two degrees of freedom and whose fundamental relation is homogeneous of order one. 	
	\end{abstract}

\maketitle


\section{Introduction}
The use of geometrical methods in theoretical physics has proven to be remarkably fruitful. The lessons we have learnt from the geometric description of field theories, and in particular from Einstein's theory of general relativity, have taught us about the central role played by curvature in describing the interaction of the fields. It is with such a spirit that we wish to review here some of the recent advances in the formulation of a geometric theory of thermodynamics or, as coined by Quevedo \cite{QuevedoGTD}, Geometrothermodynamics (GTD). In this formalism, one would like to extrapolate these ideas and make a statement of the form
	\begin{equation*}
	\label{definteraction}
	 \text{thermodynamic interaction} \approx  \text{curvature}  .
	\end{equation*}
Moreover, exploring the symmetries of a given system is an extremely powerful tool in finding its
dynamical equations of motion and allows us to explore their solutions in a clearer manner.

From the early work of Gibbs and Charatheodory \cite{gibbs, cara} to the most recent  theories based upon the use of Hessian 
metrics, the challenge of finding an intrinsic geometric formulation of thermodynamics has been pursued for almost a century. 
Hessian metrics were first applied  in thermodynamics by Rao \cite{rao45}, in 1945, using the entropy as thermodynamic potential. 
Rao's original work has been followed up and extended by a number of
authors (see, e.g., \cite{amari85} for a review).
In the latter metric efforts, Legendre invariance has not been treated with due care.  Legendre transformations (LT) only account for an interchange between conjugate pairs of thermodynamic variables once a representation has been chosen (e.g. the internal energy). This leads us  to different but equivalent descriptions through the distinct thermodynamic potentials (i.e. the enthalpy and the Helmholtz and Gibbs free energies). Just as in general relativity the physical reality cannot depend on a particular choice of coordinates, thermodynamics should be independent of the potential one uses to describe a given system. Therefore,  Legendre invariance should be an essential ingredient of any geometric construction of thermodynamics.
 
The original approach of Gibbs and  Caratheodory indicates us that the thermodynamic phase space  is endowed with a contact structure, which encodes the first law of thermodynamics (c.f. section \ref{thermo}). The GTD programme aims to promote the contact geometry of the phase space into a Riemannian contact manifold such that the metric is invariant under LT. This automatically translates into a Legendre invariant metric description in its maximal integral sub-manifold, which we identify as the space of equilibrium states. In addition, thermodynamics should also be independent of the representation one employs to describe a system, i.e. one should be able to work with the internal energy or  entropy representations, indistinctively. This can be understood as an additional symmetry  that should be incorporated in order to obtain a completely invariant geometric theory of thermodynamics.
 
The invariance under a change of representation has never been considered within the GTD formalism.  The main aim of this work is to address this issue through the construction of a metric guided by its underlying symmetries. We show that the outcome of this process leads to a particular form of one of the previously known families of metrics in GTD. Thus, obtaining a consistent result from a different point of view. 

The paper is organised as follows. In section \ref{thermo}, we make a brief review of the geometric structure of thermodynamics. In section \ref{GTD}, we present the standard results known from the GTD programme.  In section \ref{invariance}, we analyse  the behaviour of the thermodynamic metrics under a change of representation. Finally, in section \ref{conclusions}, we write the conclusions and further perspectives.

\section{Geometric structure of thermodynamics.}
\label{thermo}

In standard equilibrium thermodynamics \cite{Callen}, a system with $n$ degrees of freedom is fully specified by $n$ extensive variables $E^a$ together with their corresponding conjugate intensive variables $I_a$ and a thermodynamic potential $\Phi$ relating them. To geometrise such a system, one considers two elements:
	\begin{enumerate}
		\item A $2n+1$ dimensional  manifold $\mathcal{T}$, endowed with a contact structure $\xi\subset T\mathcal{T}$, that is, a maximally non-integrable family of hyperplanes  
	\beq
	\label{gtd.hyper}
	\xi=\ker(\Theta) 
	\eeq	
defined through some 1-form $\Theta$ satisfying the non-integrability condition
	\beq
	\label{gtd.ni}
	\Theta \wedge (\d \Theta)^n \neq 0,
	\eeq
and 
	\item the \emph{Legendre embedding}
	\beq
	\label{gtd.emb}
	\varphi:\mathcal{E}\longrightarrow\mathcal{T},
	\eeq
where $\mathcal{E}\subset\mathcal{T}$ is the $n$-dimensional integral sub-manifold defined by the \emph{isotropic condition}
	\beq
	\label{gtd.01}
	\varphi^*(\Theta) = 0.
	\eeq
	\end{enumerate} 
We call the  contact manifold $\mathcal{T}$ and the Legendre sub-manifold  $\mathcal{E}$ the \emph{thermodynamic phase space} and  the \emph{space of equilibrium states}, respectively. The reason for these names follows from the isotropic condition, equation \eqref{gtd.01},  which can be interpreted as a the geometric statement of the first law of thermodynamics. To see  this, note that one can always find a set of local coordinates $Z^A = (\Phi, I_1, ..., I_n,E^1, ..., E^n )$ for $\mathcal{T}$, such that the  contact $1-$form $\Theta$ is written as
	\beq
	\label{gtd.local}
	\Theta = \d \Phi - I_a \d E^a \quad \text{(Darboux theorem)},
	\eeq 
where we will use Einstein's summation convention unless explicitly stated otherwise, capital indices range from $0$ to $2n$ and lower case indices run from $1$ to $n$. In this coordinate system, the embedding \eqref{gtd.emb} takes the form
	\beq
	\varphi(E^a) = (\Phi(E^a), E^a, I_a),
	\eeq
and equation \eqref{gtd.01} simply becomes
	\beq
	\label{gtd.02}
	\varphi^*(\Theta) = \varphi^*\left(\d \Phi - I_a \d E^a\right) = \left(\frac{\partial \Phi}{\partial E^a} - I_a\right)\d E^a  = 0.
	\eeq
It follows immediately that
	\beq
	\label{gtd.03}
	 \frac{\partial \Phi}{\partial E^a} = I_a.
	\eeq

Equations \eqref{gtd.02} and \eqref{gtd.03} constitute the standard Gibbs relation of equilibrium thermodynamics in $\mathcal{E}$, i.e., 
	\beq
	\label{firstlaw}
	\d \Phi = I_a \d E^a,
	\eeq
where $\Phi$ represents the \emph{thermodynamic fundamental relation} for the set of state variables $E^a$. We call the contact 1-form $\Theta$ expressed in the coordinates $Z^A$ the \emph{Gibbs 1-form}.

Note that the contact 1-form $\Theta$ is not unique. Indeed, any other 1-form defining the same family of hyperplanes, equation  \eqref{gtd.hyper}, is necessarily  conformally equivalent to $\Theta$. Thus, a contact structure is an \emph{equivalence class} $[\Theta]$ of 1-forms satisfying \eqref{gtd.ni} related by a conformal transformation, i.e. for any two 1-forms $\Theta_1$ and $\Theta_2$ in $[\Theta]$
	\beq
	\label{gtd.scale}
	\Theta_1 \sim \Theta_2 \quad \text{iff} \quad \Theta_2 = \Omega \Theta_1
	\eeq
for some real function $\Omega$. We further require that the contact structure is oriented, thus $\Omega$ must have a definite sign. As we will shortly show, each member of the class corresponds to a different thermodynamic representation (c.f. section \ref{CRsec}, below). To this end, let us consider the class of maps which leaves invariant the contact structure. Let $f:\mathcal{T} \longrightarrow \mathcal{T}$ be a diffeomorphism of the thermodynamic phase-space. If $f$ preserves the contact structure, i.e
	\beq
	f^*(\Theta) = \Omega(Z^A) \Theta = \Theta' \quad \text{where} \quad \Omega\neq 0,
	\eeq
we call $f$ a \emph{contactomorphism} \cite{banyaga}. Note that $\Theta'\in[\Theta]$.

\subsection{Legendre transformations}
\label{lenendresec}

 A Legendre transformation corresponds to a redefinition of the thermodynamic potential by exchanging the role played by conjugate pairs of extensive and intensive variables. Note that such a role is physically relevant while working on the space of equilibrium states $\mathcal{E}$, whilst it is  a  mere change of coordinates on the phase space $\mathcal{T}$. These transformations are defined through the relations
	\begin{align}
	\label{LT.1}
	\tilde \Phi_{(i)} 	&\equiv  \Phi - I_{(i)} E^{(i)} \quad (\text{no sum over }i),\\
	\tilde I_{(i)} 	&\equiv  E^{(i)} \quad  \text{and}\\
	\label{LT.2}	
	\tilde E^{(i)}	&\equiv  -I_{(i)},
	\end{align}
while leaving the rest of the coordinates unchanged, i.e. $\tilde I_j = I_j$ and $\tilde E^j= E^j$ for $j\neq i$. 

The transformation \eqref{LT.1} - \eqref{LT.2} is called a \emph{partial Legendre transformation} (PLT) since it only interchanges the $i$th pair of thermodynamic variables. Thus, the transformation that exchanges every pair of coordinates called the \emph{total Legendre transformation} (TLT).

Note that Legendre transformations belong to the special class of contactomorphisms which leave the representative unchanged, i.e. $f^*(\Theta) = \Theta$. 

\subsection{Change of representation}
\label{CRsec}

In general, we are also interested in obtaining a description of a single thermodynamic system in terms of a different potential, which does not correspond to a Legendre transformation in the prescribed sense. In this case, we are truly looking for a different fundamental representation. This entails us choosing a different 1-form in the class defining the contact structure of phase space [see \eqref{gtd.scale}], i.e. a contactomorphism $f^*(\Theta) = \Omega(Z^A) \Theta$ such that $\Omega(Z^A) \neq 1$.

There is a particular set of contactomorphisms which is of thermodynamic relevance. That is, those defined by interchanging the thermodynamic potential with one of the extensive variables. If we express $\Theta$ in Darboux coordinates [equation \eqref{gtd.local}], the transformation exchanging $\Phi$ with the $i$th extensive variable is 
	\beq
	\label{newrep}
	f_{(i)}^*(\Theta) = \Theta_{(i)} = -\frac{1}{I_{(i)}}\Theta = \d E^{(i)} - \frac{1}{I_{(i)}}\d \Phi + \sum_{j\neq i}\frac{I_j}{I_{(i)}}\d E^j,
	\eeq
and using Darboux's theorem, we can write $\Theta_{(i)}$ in its canonical form
	\beq
	\Theta_{(i)} = \d \Phi' - I_{a'} \d E^{a'},
	\eeq
from which we can read the change of coordinates as
	\beq
	\label{CR2}
	\Phi' = E^{(i)}, \quad 
	E^{(i)'}  =  \Phi, \quad
	E^{j'} =  E^j, \quad
	I_{(i)'}  =  \frac{1}{I_{(i)}} \quad \text{and} \quad
	I_{j'} =  -\frac{I_j}{I_{(i)}}.
	\eeq

It is clear that the isotropic condition - equation \eqref{gtd.01} - defining the first law, is invariant under the change of scale 
	\beq
	\varphi^*\left(-\frac{1}{I_{(i)}}\Theta\right) = \varphi^*\left(-\frac{1}{I_{(i)}}\right)\varphi^*(\Theta) = 0.
	\eeq
Therefore, we obtain the equilibrium relations in the new coordinates
	\beq
	\label{CR.newfund}
	E^{(i)} = E^{(i)}(\Phi, E^j), \quad \frac{\partial E^{(i)}}{\partial \Phi} = \frac{1}{I_{(i)}} \quad \text{and} \quad \frac{\partial E^{(i)}}{\partial E^j} = -\frac{I_j}{I_{(i)}} \quad (j\neq i),	
	\eeq
and the first law in the $E^{(i)}$ representation is simply
	\beq
	\d E^{(i)} = \frac{1}{I_{(i)}}\d\Phi - \sum_{j\neq i}\frac{I_j}{I_{(i)}}\d E^j.
	\eeq

In sum, the symmetries we will demand in the forthcoming sections are motivated by the fact that the contact 1-form $\Theta$ is invariant under Legendre transformations, whilst a change of representation corresponds to selecting a different 1-form in the class \eqref{gtd.scale} defining the same contact structure. Note that these `symmetries' leave invariant the space of equilibrium states $\mathcal{E}$ and they will be automatically inherited by its geometric properties.


\section{The Riemannian structure of the GTD programme}
\label{GTD}

In addition to the geometric description of thermodynamics in terms of a contact structure, the GTD programme promotes the contact manifold $(\mathcal{T},[\Theta])$ into a Riemannian contact manifold $(\mathcal{T},[\Theta],G)$, where $G$ is a metric sharing the symmetries of $\Theta$. The class of metrics satisfying this requirement is vast and there is currently no general principle to select a particular one. The way to deal with this ambiguity has been to introduce some physical input from known systems and demanding that the curvature associated with the induced metric in the space of equilibrium states accounts for the expected phenomena, i.e. to be zero in the case of the ideal gas, or to diverge as one approaches a phase transition, taking us away from local equilibrium hypothesis.  

Thus far, there are two independent families of metrics for $\mathcal{T}$ which can be classified according to their invariance properties \cite{fundamentals}, that is, those which are invariant under total Legendre transformations only, and those which are also invariant under partial Legendre transformations. Let us write them as
	\beq
	\label{quevedoGII}
	G_{\rm T}  = \Theta \otimes \Theta + \Lambda(Z^A)\left (\xi^a_{\ b}E^b I_a \right) \left(\chi^c_{\ d}\ \d E^d \otimes \d I_c \right)
	\eeq
and
	\beq
	\label{quevedoGIII}
	G_{\rm P}  =\Theta \otimes \Theta +\Lambda\left(Z^A\right) \sum_{i=1}^n \left[\left(E^i I_i \right)^{2k+1} \d E^i \otimes \d I_i \right],
	\eeq
where $\Lambda(Z^A)$ is an arbitrary Legendre invariant function of the coordinates $Z^A$, $k$ is an integer and $\xi^a_{\ b}$ and $\chi^a_{\ b}$ are diagonal constant matrices. Note that these matrices are not tensors. Their purpose is solely to indicate the form of the metrics. In a previous work \cite{phasetransitions}, the specific form has been determined through the correct description of the relevant physical phenomena. On the one hand, the choice $\xi^a_{\ b} = \delta^a_{\ b}$ and $\chi^a_{\ b} = \delta^a_{\ b}$, has been used to describe systems with first order phase transitions. On the other hand, second order phase transitions in black holes have been correctly described when $\xi^a_{\ b} = \delta^a_{\ b}$ and $\chi^a_{\ b} = \eta^a_{\ b}$, where $\eta^a_{\ b} = \text{diag}[-1,1,\ldots,1]$. These metrics are known in the GTD literature as $G_I$ and $G_{II}$, respectively. Finally, we use the labels T and P to denote invariance under total and partial Legendre transformations. 

The corresponding induced metrics in the space of equilibrium states are simply given by
	\begin{equation}
		\label{quevedogI}
	g_{\rm T} = \varphi^*(G_{\rm T}) = \Lambda \left(\xi^a_{\ b} E^b \frac{\partial \Phi}{\partial E^a} \right) \chi^c_{\ d} \frac{\partial^2 \Phi}{\partial E^c \partial E^e}\ \d E^d \otimes \d E^e ,
	\end{equation}
and
	\begin{equation}
	\label{quevedogIII}
	g_{\rm P}=\varphi^*(G_{P}) = \Lambda \sum_{i,j=1}^n\left[\left(E^i \frac{\partial \Phi}{\partial E^i} \right)^{2 k + 1}  \frac{\partial^2 \Phi}{\partial E^i \partial E^j}\ \d E^i \otimes \d E^j\right].
	\end{equation}

The core idea in the GTD programme is that the curvature associated with either \eqref{quevedogI} or \eqref{quevedogIII} contains all the information about  the `thermodynamic interaction' of a system specified by its fundamental relation $\Phi(E^a)$.  For example, the lack of thermal interaction of the ideal gas is reflected by  the vanishing of  its associated curvature scalar. Similarly, this approach has proven to describe accurately the critical behaviour of various systems as curvature singularities, i.e. the configurations where the local equilibrium hypothesis is no longer valid
such as, e.g., the van der Waals gas \cite{phasetransitions}.

Despite the fact that both families of metrics for $\mathcal{E}$, equations \eqref{quevedogI} and \eqref{quevedogIII},  are induced from the manifestly Legendre invariant metrics on $\mathcal{T}$, equations \eqref{quevedoGII} and \eqref{quevedoGIII}, in general,  do not produce the same curvature for $\mathcal{E}$ when one changes from one fundamental representation to another (c.f. section \ref{CRsec}, above). The physical outcome of choosing a different member of the class $[\Theta]$ generating the contact structure of $\mathcal{T}$ should leave the geometric properties of $\mathcal{E}$ unchanged. Therefore, we will demand the CR invariance of the programme through the isometry of the metric on $\mathcal E$.


\section{The change of representation in GTD}
\label{invariance}

The change of representation has remained a largely unanalysed issue in previous work on GTD. In this section, we consider the change of representation as described in section \ref{CRsec} and we analyse how the induced metrics behave under such transformation. To this end, let us note the following points:
	\begin{enumerate}
	\item Our analysis will only consider systems which are described by homogeneous functions of a definite order.
	\item If we have a homogeneous fundamental relation $\Phi(\lambda E^a) = \lambda^\beta \Phi(E^a)$, the change of representation $E^{(i)} = E^{(i)}(\Phi,E^j)$ with $j\neq i$ is not a homogeneous function.
	\item The representation in which the system is described by a homogeneous function will be called the \emph{canonical representation} and we will label it by $\Phi$.
	\end{enumerate}
 In this sense, the phase-space metrics, \eqref{quevedoGII} and \eqref{quevedoGIII}, together with their corresponding induced metrics \eqref{quevedogI} and \eqref{quevedogIII},  are written in the canonical representation.
 
 Consider a  slight generalization of the  metric $G_{\rm T}$,   
 	\begin{align}
	\label{cesarGII}
	G^\Phi &= \Theta \otimes \Theta + \left (\xi^a_{\ b}E^b I_a \right)\  \sum_{k,d}  \Lambda_k(Z^A) \left(\chi^k_{\ d}\ \d E^d \otimes \d I_k \right),
	\end{align}
where $\xi^a_{\ b}$ and $\chi^a_{\ b}$ are as in \eqref{quevedoGII}. Choosing the different representative $\Theta_{(i)}$ [c.f. equation \eqref{newrep}], we can rewrite $G^\Phi$ as
	\begin{align}
	\label{cesarGIIb}
	G^{E^{(i)}} &= \Theta_{(i)} \otimes \Theta_{(i)} + \left (\xi^{a'}_{\ b'}E^{b'} I_{a'} \right)\  \sum_{k',d'}  \Lambda_{k'}(Z^{A'}) \left(\chi^{k'}_{\ d'}\ \d E^{d'} \otimes \d I_{k'} \right),
	\end{align}
which, using equations (\ref{CR2}), can be related to the un-primed coordinates as
	\begin{align}
	\label{quevedoGIIi}
	G^{E^{(i)}} &= \frac{1}{I_{(i)}^{\ \ 2}} \Theta \otimes \Theta + \left[\xi^{(i)}_{\ (i)} \frac{\Phi}{I_{i}} - \sum_{j\neq i} \xi^j_{\ j} \frac{E^j I_j}{I_{(i)}}  \right] \left[\Lambda_{(i)}\chi^{(i)}_{\ (i)} \d \Phi \otimes \d\left(\frac{1}{I_{(i)}}\right) + \sum_{j\neq i}\Lambda_j \chi^j_{\ j} \d E^j \otimes \d \left(-\frac{I_j}{I_{(i)}} \right) \right].
	\end{align}
There is an implicit change in the $\Lambda$-functions under the prescribed coordinate transformation, namely
		\beq
		\Lambda_{(i)} =\Lambda_{(i)}\left[Z^{A'}\left(Z^A\right)\right] \quad \text{and} \quad \Lambda_{j} = \Lambda_j\left[Z^{A'}\left(Z^A\right)\right]. 
		\eeq

Now we will prove the following
	\begin{lemma}
	If $\Lambda_k$ is a Legendre invariant function for all $k$, then $G^\Phi$ is invariant under TLTs.
	\end{lemma}
	\begin{proof}
	It follows from the invariance of $G_{\rm T}$ provided each of the $\Lambda_k$ is itself invariant.
	\end{proof}

	\begin{prop}
	Let the fundamental relation $\Phi = \Phi(E^a)$ be a homogeneous function of order $\beta$.	Then, the induced metrics $g^\Phi = \varphi^*(G^\Phi)$ and $g^{E^{(i)}} = \varphi^*(G^{E^{(i)}})$ are conformally related if and only if $\Lambda_{(i)} = \Lambda_j \ \chi^j_{\ j}\left(\chi^{(i)}_{\ (i)}\right)^{-1}$  (no sum over $j$) for all $j\neq i$.
	\end{prop}
	\begin{proof}
	
	The induced metric $g^\Phi = \varphi^*(G^\Phi)$ can formally be written as 
	\beq  
	\label{gIPhi}
 	g^{\Phi} =  (\xi^{a}_{\ b}  I_a  E^b )\  \sum_k \Lambda_k \, \chi^{k}_{\ c}\  \d I_k \otimes  \d E^c,
	\eeq
where $I_a=\partial \Phi/\partial E^a$ and therefore 
	\beq \label{dI}
    \d I_k = \frac{\partial^2 \Phi}{\partial E^k \partial E^b} \ \d E^b.
	\eeq

Now, the induced metric in the $E^{(i)}$ representation, $g^{E^{(i)}}=\varphi^*(G^{E^{(i)}})$ is
	\begin{align}
	\label{quevpullback2}
	g^{E^{(i)}} & = \left[\xi^{(i)}_{\ (i)} \frac{\Phi}{I_{(i)}} - \sum_{j\neq i} \xi^j_{\ j} \frac{E^j I_j}{I_{(i)}}  \right] \left[\Lambda_{(i)} \chi^{(i)}_{\ (i)}\left(\frac{\partial^2 E^{(i)}}{\partial \Phi^2} \d \Phi \otimes \d \Phi + \sum_{j\neq i} \frac{\partial^2 E^{(i)}}{\partial E^j \partial \Phi} \d E^j \otimes \d \Phi \right) \right.\nonumber \\
					& \left. +\sum_{j\neq i} \Lambda_j \chi^j_{\ j} \left(\frac{\partial^2 E^{(i)}}{\partial \Phi \partial E^j} \d \Phi \otimes \d E^j  +\sum_{k\neq i} \frac{\partial^2 E^{(i)}}{\partial E^j \partial E^k} \d E^j \otimes \d E^k \right) \right].
	\end{align}
In this representation we have the analogous relations to \eqref{dI} 
	\begin{align} 
	\label{relationsIa}
	\d \left(\frac{1}{I_{(i)}}\right) &= \frac{\partial^2 E^{(i)}}{\partial \Phi^2} \d \Phi + \sum_{k\neq i}\frac{\partial^2 E^{(i)}}{\partial E^k \partial \Phi} \d E^k, \\
	\label{relationsIb}
	\d \left(-\frac{I_j}{I_{(i)}}\right) &= \frac{\partial^2 E^{(i)}}{\partial \Phi \partial E^j} \d \Phi + \sum_{k\neq i}\frac{\partial^2 E^{(i)}}{\partial E^k \partial E^j} \d E^k.
	\end{align}
Since $\Phi$ is a homogeneous function of order $\beta$, we have that
	\beq  
	\label{genGD}
	\beta \Phi = I_a E^a.
	\eeq
Using this result, together with equations \eqref{relationsIa}, \eqref{relationsIb} and the first law of thermodynamics, equation \eqref{firstlaw}, the expression for the induced metric \eqref{quevpullback2} becomes
	\begin{align}
	\label{gIEAbetaAux}
 	g^{E^{(i)}} &=  -\frac{1}{\beta I_{(i)}}  \left[ \xi^{(i)}_{\ (i)} E^{(i)}  +  \sum_{j\neq i}\left(\xi^{(i)}_{\ (i)} - \xi^j_{\ j}\beta \right) \frac{I_jE^j}{I_{(i)}}\right] \nonumber\\ 
 	&\times\left[ -\Lambda_{(i)} \chi^{(i)}_{\ (i)} \frac{1}{I_{(i)}} \d E^{(i)} \otimes \d I_{(i)} - \Lambda_{(i)} \chi^{(i)}_{\ (i)} \sum_{j\neq i} \frac{I_j}{I_{(i)}^2} \d E^j \otimes \d I_{(i)}  \right. \nonumber \\
&- \left. \sum_{j\neq i} \Lambda_j \chi^j_{\ j} \frac{1}{I_{(i)}} \d E^j \otimes \d I_j + \sum_{j\neq i} \Lambda_j \chi^j_{\ j} \frac{I_j}{I_{(i)}^2} \d E^j \otimes \d I_{(i)} \right],
	\end{align}
which can be factorised to
\begin{align}
	\label{gIEAbeta}
 	g^{E^{(i)}} &= -\frac{1}{\beta I_{(i)}}  \left[ \xi^{(i)}_{\ (i)} E^{(i)}  +  \sum_{j\neq i}\left(\xi^{(i)}_{\ (i)} - \xi^j_{\ j}\beta \right) \frac{I_jE^j}{I_{(i)}}\right]\nonumber \\ 
						& \times \left[ \sum_{k} \Lambda_k \chi^k_{\ c}  \d I_k \otimes \d E^c  + \sum_{j\neq i} \left(\Lambda_j \chi^j_{\ j} -  \Lambda_{(i)} \chi^{(i)}_{\ (i)}\right) \frac{I_j}{{I_{(i)}}^2} \d E^j \otimes \d I_{(i)} \right].
	\end{align}

It follows that the two metrics are conformally related only when the condition
	\beq
 	\label{cond1}
	\Lambda_{(i)} = \Lambda_j \ \frac{\chi^j_{\ j}}{\chi^{(i)}_{\ (i)}} \quad \text{no sum over }j \quad \forall j\neq i,
	\eeq
is satisfied. In such case, and using \eqref{gIPhi}, equation \eqref{gIEAbeta} reduces to
	\beq
	\label{gtotconf}
	g^{E^{(i)}} = -\frac{1}{\beta I_{(i)}}  \left[ \xi^{(i)}_{\ (i)} E^{(i)}  +  \sum_{j\neq i}\left(\xi^{(i)}_{\ (i)} - \xi^j_{\ j}\beta \right) \frac{I_jE^j}{I_{(i)}}\right] \left[\xi^a_{\ b} E^b I_a \right]^{-1} g^{\Phi}.
	\eeq
Hence, the induced metrics in the two representations are conformally related. 
	\end{proof}

In the case of the GTD programme, condition \eqref{cond1} together with $\Lambda_{(i)} = \Lambda_{j} = \Lambda$, yield the metric determined by $\chi^a_{\ b} = \delta^a_{\ b}$, namely $G_I$. Notice that, the same condition rules out the choice $\chi^a_{\ b} = \eta^a_{\ b}$, that is, $G_{II}$ does not lead to conformally related metrics in $\mathcal{E}$ for different representations.

 	\begin{prop}
	The induced metric is invariant under change of representation if the conformal factor is
	\beq
	\label{confactor}
	\Lambda(Z^A) = \frac{1}{\xi^a_{\ b} E^b I_a}\sum_{j\neq i}\frac{1}{ E^j I_j}
	\eeq
	\end{prop}
	
	\begin{proof}
	The metric $G^\Phi$ \eqref{cesarGII} with the choice \eqref{confactor} gives,
	\beq
	G^\Phi = \Theta \otimes \Theta + \sum_{j\neq i} \frac{1}{E^j I_j} \sum_{k,d} { \Lambda_k(Z^A)} \left(\chi^k_{\ d}\ \d E^d \otimes \d I_k \right),
	\eeq
whilst 
	\beq
	G^{E^{(i)}} = \frac{1}{{I_{(i)}}^{2}} \Theta \otimes \Theta + \sum_{j\neq i}\frac{-I_{(i)}}{E^jI_j} \left[\d \Phi \otimes \d \left(\frac{1}{I_{(i)}}\right) + \sum_{j\neq i} \d E^j \otimes \d \left(-\frac{I_j}{I_{(i)}}\right)\right].
	\eeq

The pulled-back metrics are
	\beq
	\label{invariantmetric}
	g^\Phi = \sum_{j\neq i}\frac{1}{E^jI_j} \d E^a \otimes \d I_a,
	\eeq
and analogously to \eqref{gIEAbeta} we can factorize $\varphi^*\left(G^{E^{(i)}}\right)$ to obtain,
	\beq
	g^{E^{(i)}} = \sum_{j\neq i} \frac{-I_{(i)}}{E^jI_j} \left( -\frac{1}{I_{(i)}} \d E^a \otimes \d I_a \right).
	\eeq

It follows immediately that
	\beq
	g^\Phi = g^{E^{(i)}}.
	\eeq

	\end{proof}

Note that the $\Lambda$-function is related to the primed coordinates through
	\beq
	\Lambda\left[Z^{A'}\left(Z^A\right)\right] = \frac{-\beta I_{(i)}}{\left[ \xi^{(i)}_{\ (i)} E^{(i)}  +  \sum_{j\neq i}\left(\xi^{(i)}_{\ (i)} - \xi^j_{\ j}\beta \right) \frac{I_jE^j}{I_{(i)}}\right]}\ \sum_{j\neq i}\frac{1}{ E^j I_j}
	\eeq
	
	\begin{corolary}
	$\Lambda(Z^A)$ is invariant under total Legendre transformations.
	\end{corolary}	
	\begin{proof}
	Using \eqref{LT.1}-\eqref{LT.2} in \eqref{confactor} for every pair of indexes we obtain
	\beq
	\tilde{\Lambda} = \frac{-1}{\xi^a_{\ b}\ \tilde{E}^b \tilde{I}_a} \sum_{j\neq i} \frac{-1}{\tilde{E^j}\tilde{I_j}} = \frac{1}{\xi^a_{\ b}(\tilde{E}^b \tilde{I}_a)} \sum_{j\neq i} \frac{1}{\tilde{E^j}\tilde{I_j}}=\frac{1}{\xi^a_{\ b} E^b I_a}\sum_{j\neq i}\frac{1}{ E^j I_j} = \Lambda.
	\eeq

	\end{proof}
	
Therefore, we have obtained a metric which is invariant under total Legendre transformations and whose associated curvature in the space of equilibrium states does not depend upon the chosen fundamental representation, provided it is a homogeneous function. In the canonical representation it is written as
	\beq
	\label{gnatural}
	G^\natural_{\rm T} = \Theta \otimes \Theta + \sum_{j\neq i}\frac{1}{{E^j} I_{j}} \d E^a \otimes \d I_a.  
	\eeq
	
Now, let us show that invariance under change of representation is not compatible with $G_{\rm P}$. To this end, let us write the metric \eqref{quevedoGIII} in the $E^{(i)}$ representation, that is
	\beq
	G_{\rm P}^{E^{(i)}} = \frac{1}{{I_{(i)}}^2} \Theta \otimes \Theta + \Lambda \left[\left(\frac{\Phi}{I_{(i)}} \right)^{2k+1} \d \Phi \otimes \d \left(\frac{1}{I_{(i)}} \right) + \sum_{j \neq i} \left(\frac{E^j I_j}{I_{(i)}} \right)^{2k+1} \d E^j \otimes \d \left(-\frac{I_j}{I_{(i)}} \right)\right]. 
	\eeq
	Thus, we can prove the following
	\begin{prop}
	Let $\Phi = \Phi(E^a)$ be a homogeneous function of order $\beta$,  $G^{\Phi}_{\rm P}$ and $G^{E^{(i)}}_{\rm P}$ the metric \eqref{quevedoGIII} in the canonical and $E^{(i)}$ representations, respectively. Then, the induced metrics $g_{\rm P}^{\Phi} = \varphi^{*}\left(G^{\Phi}_{\rm P}\right)$ and $g_{\rm P}^{E^{(i)}} = \varphi^*\left( G^{E^{(i)}}_{\rm P}\right)$ cannot be conformally related.    
	\end{prop}
	\begin{proof}
	The induced metrics are
	\begin{align}
	\label{gpi}
	g_{\rm P}^{\Phi}    & = \Lambda \sum_{i=1}^n \left(E^i \frac{\partial \Phi}{\partial E^i} \right)^{2k+1} \d E^{i} \otimes \d I_{i} \quad \text{and}\\
	\label{gpei}
	g_{\rm P}^{E^{(i)}} & = \Lambda \left[\left(\frac{\Phi}{I_{(i)}} \right)^{2k+1}\  \d \Phi \otimes \d \left(\frac{1}{I_{(i)}}\right) + \sum_{j \neq i}\ \left(\frac{E^j I_j}{I_{(i)}} \right)^{2k+1} \d E^j \otimes \d \left(-\frac{I_j}{I_{(i)}} \right) \right],
	\end{align}
	where the differentials of the intensive variables are the same as in \eqref{dI}, \eqref{relationsIa} and \eqref{relationsIb}.
	
	Using the generalised Euler identity \eqref{genGD}, we can rewrite \eqref{gpei} as
	\begin{align}
	\label{geiinter}
	g^{E^{(i)}}_{\rm P} = & \Lambda \left[\frac{-1}{{I_{(i)}}^{2}}\left(\frac{E^a I_a}{\beta I_{(i)}} \right)^{2k+1} \d \Phi \otimes \d I_{(i)}  + \sum_{j\neq i} \left(\frac{E^j I_j}{I_{(i)}} \right)^{2k+1} \d E^j \otimes \left(-\frac{1}{I_{(i)}} \d I_j + \frac{I_j}{{I_{(i)}}^2}\d I_{(i)} \right)\right],
	\end{align}
	and substituting the first law \eqref{firstlaw}, we can factorise the expression above to obtain
	\begin{align}
	\label{geifinal}
	g^{E^{(i)}}_{\rm P} = & -\Lambda \left[\left(\frac{E^a I_a}{\beta I_{(i)}} \right)^{2k+1}  \sum_{j\neq i} \frac{I_j}{{I_{(i)}}^2} \left[1 - \left(\beta \frac{E^j I_j}{E^a I_a}\right)^{2k+1} \right] \d E^j \otimes \d I_{(i)}\right.\nonumber\\
						  & \left. + \frac{1}{I_{(i)}}\left(\frac{E^a I_a}{\beta I_{(i)}} \right)^{2k+1} \d E^{(i)} \otimes \d I_{(i)} + \frac{1}{I_{(i)}}\sum_{j\neq i} \left(\frac{E^j I_j}{I_{(i)}} \right)^{2k+1} \d E^j \otimes d I_j \right]. 
	\end{align}
	
	The only possibility of making \eqref{geifinal} conformal to \eqref{gpi} is that $k=-1/2$, which is inconsistent with  the partial Legendre invariance of $G_{\rm P}$.
	\end{proof}

\subsection{Examples: homogeneous systems with two degrees of freedom}

In the simplest situation, when the fundamental relation is homogeneous of order one, i.e. when $\beta =1$ {and $\xi^{(i)}_{\ (i)} = \xi^j_{\ j}=1$ for any $i\neq j$}, the metric $g_{\rm T}^{E^{(i)}}$, equation \eqref{gtotconf}, reduces to
	\beq
	\label{gIEA}
 	g^{E^{(i)}} = -\left[I^{-1}_{(i)} E^{(i)} 
        \frac{1}{ I_a E^a}\right] \ g^{\Phi}.
	\eeq
Now, let us consider a system with two degrees of freedom. The two representations that are commonly used are those of the  energy and entropy.  Let us take $\Phi = U(S,V)$ and $E^{(i)} = S(U,V)$. In this case, the induced metrics are conformally related as
	\beq
	\label{gIconf}
 	g^{S} = -\left[T^{-1} S\left(\frac{1}{ST -PV}\right)\right] \ g^{U}.
	\eeq
It is clear that these two conformally related metrics do not produce the same curvature. Thus, we will not obtain the same thermodynamic information whenever we make a change of representation.

Note that if we work instead with the metric (\ref{invariantmetric})
	\beq
	g^\natural = \varphi^*\left(G^\natural \right) = \sum_{j \neq i} \frac{1}{E^j I_j} \d E^a \otimes \d I_a ,
	\eeq
we obtain 
	\beq
	{g^{U}_{\rm T}}^\natural = -\frac{1}{P V} \left( \d S \otimes \d T - \d V \otimes \d {P}\right) = {g_{\rm T}^S}^\natural
	\eeq
and, therefore, the curvature scalar is the same in both representations, i.e., the change of representation is an isometry for such a metric.

Finally, let us consider the manifestly not Legendre invariant case of $G_{\rm P}$ with $\Lambda = 1$ and $k=-1/2$, whose pullback generates Hessian metrics for the equilibrium space.
In this case, the metrics $g_{\rm P}^U$ and $g_{\rm P}^S$ -- equations \eqref{gpi} and \eqref{geifinal}, respectively -- are conformally related through \cite{mrugala1,mrugala2}
	\beq
	g_{\rm P}^{S} = -\frac{1}{T} g_{\rm P}^{U}.
	\eeq
This case corresponds to the Hessian metric with the entropy as thermodynamic potential which was originally proposed by Rao \cite{rao45}.

This simple exercise shows that Hessian metrics not only fail to be Legendre invariant, but they also  give different curvatures in each representation.

\section{Conclusions.}\label{conclusions}

In this paper, we have analysed in detail the invariance properties of the GTD programme. As it has been previously argued, Legendre invariance is paramount in preserving the notion that the physical reality should be independent of the thermodynamic potential used to describe it. Within the GTD programme, analogously to field theories, curvature is the geometric object accounting for such reality. The metrics $G_{\rm T}$ and $G_{\rm P}$ -- equations \eqref{quevedoGII} and \eqref{quevedoGIII}, respectively -- satisfy the desired invariance in the thermodynamic phase-space $\mathcal{T}$ and, therefore, produce the same curvature for the space of equilibrium states $\mathcal{E}$ independently of the thermodynamic potential used. However, Legendre invariance alone is not sufficient to guarantee a unique description of a thermodynamic system in terms of its curvature, i.e. we also need to demand invariance of the curvature under a change of fundamental representation. Such a problem has remained largely unanalysed.

We have shown that the only metric compatible with the invariance under change of representation in the GTD programme is $G_{\rm T}$ with $\Lambda$ given by \eqref{confactor}, that is $G^\natural_{\rm T}$ [c.f. equation \eqref{gnatural}]. This metric, in turn, is a particular case of the metric family $G_{I}$ which has been used in GTD to describe thermodynamic systems with first-order phase transitions.
Consequently, the metric $G^\natural_{\rm T}$ cannot be applied to systems with second-order phase transitions which are described with the metric $G_{II}$.
It is also important to stress the fact that the results presented in this work only apply to systems which can be described by homogeneous functions. We have also shown that invariance under partial Legendre transformations cannot be preserved if we demand representation invariance for $G_{\rm P}$. 

Finally, we have applied our results to a generic homogeneous system with two degrees of freedom using the energy and entropy representations. We observe that indeed, $G^\natural_{\rm T}$ is the only of the metrics which gives an invariant curvature in both representations, whereas $G_{\rm P}$ with $\Lambda = 1$ and the Legendre invariance violating condition $k = -1/2$ reduces to a particular Hessian metric in which the entropy is used as thermodynamic potential. 

All the metrics found so far in GTD, using only the Legendre invariance condition, contain the arbitrary conformal factor $\Lambda(Z^A)$. 
This is an additional degree of freedom that can be used to reach diverse objectives. For instance, in the study of the thermodynamics of black holes \cite{tqs12}, we found that the curvature singularities determine the phase transition structure and, in addition, $\Lambda$ can be chosen in such a way that the limiting case of extremal black holes corresponds to curvature singularities too. Here, we have found that $\Lambda$ can also be used to reach representation invariance directly from the phase space. We believe that the conformal freedom that follows from the phase 
space still might have more applications at the level of the equilibrium space.

This work will serve as a solid footing to explore further possibilities in our quest to obtain a completely invariant description of thermodynamics. It cannot be over estimated that the results presented here depend heavily on the homogeneity of the fundamental relations describing physical systems.

\section*{Acknowledgements}

We would like to thank the members of the GTD-group for fruitful comments and discussions.
The work of AB was supported by an ICRANet fellowship. CSLM is thankful to CONACYT, Grant No. 290679\_UNAM. FN acknowledges support from DGAPA-UNAM, and HQ wishes to thank CONACYT, Grant No. 166391.

\end{document}